\def\UseBibLatex{1}%

\ifx\SoCG\undefined%
\documentclass[11pt]{article}%
\providecommand{\SoCGVer}[1]{}%
\providecommand{\NotSoCGVer}[1]{#1}%
\else%
\makeatletter
\def\input@path{{lipics/}{../lipics/}}
\makeatother
\documentclass[a4paper,USenglish,cleveref,autoref,thm-restate]{socg-lipics-v2019}
\providecommand{\SoCGVer}[1]{#1}%
\providecommand{\NotSoCGVer}[1]{}%
\fi

\ifx\sarielComp\undefined%
\newcommand{\SarielComp}[1]{}
\newcommand{\NotSarielComp}[1]{#1}%
\else
\newcommand{\SarielComp}[1]{#1}%
\newcommand{\NotSarielComp}[1]{}%
\fi
\newcommand{\IfPrinterVer}[2]{#2}%

\providecommand{\BibLatexMode}[1]{}
\providecommand{\BibTexMode}[1]{#1}

\ifx\UseBibLatex\undefined%
  \renewcommand{\BibLatexMode}[1]{}
  \renewcommand{\BibTexMode}[1]{#1}
\else
  \renewcommand{\BibLatexMode}[1]{#1}
  \renewcommand{\BibTexMode}[1]{}
\fi

\BibLatexMode{%
   \usepackage[style=alphabetic,backend=biber,bibencoding=utf8]{biblatex}%
   \usepackage{sariel_biblatex}%
}

\usepackage[cm]{fullpage}%
\usepackage{amsmath}%
\usepackage{amssymb}%
\usepackage{xcolor}%
\usepackage{caption}%

\SarielComp{\usepackage{sariel_colors}}%

\NotSoCGVer{%
   \usepackage[amsmath,thmmarks]{ntheorem}%
   \theoremseparator{.}%
}

\usepackage{titlesec}%
\titlelabel{\thetitle. }%
\usepackage{xcolor}%
\usepackage{mleftright}%
\usepackage{xspace}%
\usepackage{hyperref}%
\usepackage{mathtools}%
\usepackage{stmaryrd}%

\usepackage{tabularx}%
\usepackage{array}

\IfPrinterVer{%
   \usepackage{hyperref}%
}{%
   \usepackage{hyperref}%
   \hypersetup{%
      breaklinks,%
      ocgcolorlinks, colorlinks=true,%
      urlcolor=[rgb]{0.25,0.0,0.0},%
      linkcolor=[rgb]{0.5,0.0,0.0},%
      citecolor=[rgb]{0,0.2,0.445},%
      filecolor=[rgb]{0,0,0.4},
      anchorcolor=[rgb]={0.0,0.1,0.2}%
   }
}

\SoCGVer{%
   \theoremstyle{plain}%

   \newtheorem{conjecture}[theorem]{Conjecture}
   \newtheorem{fact}[theorem]{Fact}
   \newtheorem{observation}[theorem]{Observation}
   \newtheorem{invariant}[theorem]{Invariant}
   \newtheorem{question}[theorem]{Question}
   \newtheorem{prop}[theorem]{Proposition}
   \newtheorem{openproblem}[theorem]{Open Problem}

   \theoremstyle{plain}%
   \newtheorem{defn}[theorem]{Definition}
   \newtheorem{problem}[theorem]{Problem}
   \newtheorem{xca}[theorem]{Exercise}
   \newtheorem{exercise_h}[theorem]{Exercise}
   \newtheorem{assumption}[theorem]{Assumption}%

   \newtheorem{proofof}{Proof of\!}%
}%

\NotSoCGVer{%
\theoremseparator{.}%

\theoremstyle{plain}%
\newtheorem{theorem}{Theorem}[section]

\newtheorem{lemma}[theorem]{Lemma}

\theoremstyle{plain}%
\theoremheaderfont{\sf} \theorembodyfont{\upshape}%
\newtheorem*{remark:unnumbered}[theorem]{Remark}%
\newtheorem{remark}[theorem]{Remark}%
\newtheorem{definition}[theorem]{Definition}
\newtheorem{defn}[theorem]{Definition}

\newcommand{\myqedsymbol}{\rule{2mm}{2mm}}

\theoremheaderfont{\em}%
\theorembodyfont{\upshape}%
\theoremstyle{nonumberplain}%
\theoremseparator{}%
\theoremsymbol{\myqedsymbol}%
\newtheorem{proof}{Proof:}%

}%

\newcommand{\atgen}{\symbol{'100}}
\newcommand{\SarielThanks}[1]{\thanks{Department of Computer Science;
      University of Illinois; 201 N. Goodwin Avenue; Urbana, IL,
      61801, USA; {\tt sariel\atgen{}illinois.edu}; {\tt
         \url{http://sarielhp.org/}.} #1}}

\newcommand{\TimThanks}[1]{\thanks{Department of Computer Science;
      University of Illinois; 201 N. Goodwin Avenue; Urbana, IL,
      61801, USA; {\tt tzhou28\atgen{}illinois.edu}.
      #1}}

\newcommand{\HLink}[2]{\hyperref[#2]{#1~\ref*{#2}}}
\newcommand{\HLinkSuffix}[3]{\hyperref[#2]{#1\ref*{#2}{#3}}}

\newcommand{\figlab}[1]{\label{fig:#1}}
\newcommand{\figref}[1]{\HLink{Figure}{fig:#1}}

\newcommand{\thmlab}[1]{{\label{theo:#1}}}
\newcommand{\thmref}[1]{\HLink{Theorem}{theo:#1}}

\providecommand{\deflab}[1]{\label{def:#1}}

\newcommand{\remlab}[1]{\label{rem:#1}}
\newcommand{\remref}[1]{\HLink{Remark}{rem:#1}}%

\newcommand{\seclab}[1]{\label{sec:#1}}
\newcommand{\secref}[1]{\HLink{Section}{sec:#1}}

\newcommand{\lemlab}[1]{\label{lemma:#1}}
\newcommand{\lemref}[1]{\HLink{Lemma}{lemma:#1}}%

\providecommand{\eqlab}[1]{}%
\renewcommand{\eqlab}[1]{\label{equation:#1}}

\definecolor{blue25emph}{rgb}{0, 0, 11}
\providecommand{\emphic}[2]{%
   \textcolor{blue25emph}{%
      \textbf{\emph{#1}}}%
   \index{#2}}

\providecommand{\emphi}[1]{\emphic{#1}{#1}}

\numberwithin{figure}{section}%
\numberwithin{table}{section}%
\numberwithin{equation}{section}%

\providecommand{\remove}[1]{}%

\newcommand{\pth}[2][\!]{\mleft({#2}\mright)}%

\newcommand{\ceil}[1]{\left\lceil {#1} \right\rceil}
\newcommand{\floor}[1]{\left\lfloor {#1} \right\rfloor}

\newcommand{\cardin}[1]{\left| {#1} \right|}%

\renewcommand{\th}{th\xspace}
\newcommand{\ds}{\displaystyle}%

\renewcommand{\Re}{\mathbb{R}}%

\usepackage[inline]{enumitem}

\newlist{compactenumA}{enumerate}{5}%
\setlist[compactenumA]{topsep=0pt,itemsep=-1ex,partopsep=1ex,parsep=1ex,%
   label=(\Alph*)}%

\newlist{compactenuma}{enumerate}{5}%
\setlist[compactenuma]{topsep=0pt,itemsep=-1ex,partopsep=1ex,parsep=1ex,%
   label=(\alph*)}%

\newlist{compactenumI}{enumerate}{5}%
\setlist[compactenumI]{topsep=0pt,itemsep=-1ex,partopsep=1ex,parsep=1ex,%
   label=(\Roman*)}%

\newlist{compactenumi}{enumerate}{5}%
\setlist[compactenumi]{topsep=0pt,itemsep=-1ex,partopsep=1ex,parsep=1ex,%
   label=(\roman*)}%

\newlist{compactitem}{itemize}{5}%
\setlist[compactitem]{topsep=0pt,itemsep=-1ex,partopsep=1ex,parsep=1ex,%
   label=\ensuremath{\bullet}}%

\providecommand{\Mh}[1]{#1}%

\DeclareFontFamily{U}{BOONDOX-calo}{\skewchar\font=45 }
\DeclareFontShape{U}{BOONDOX-calo}{m}{n}{
  <-> s*[1.05] BOONDOX-r-calo}{}
\DeclareFontShape{U}{BOONDOX-calo}{b}{n}{
  <-> s*[1.05] BOONDOX-b-calo}{}
\DeclareMathAlphabet{\mathcalb}{U}{BOONDOX-calo}{m}{n}
\SetMathAlphabet{\mathcalb}{bold}{U}{BOONDOX-calo}{b}{n}
\DeclareMathAlphabet{\mathbcalb}{U}{BOONDOX-calo}{b}{n}

\newcommand{\SaveContent}[2]{%
   \expandafter\newcommand{#1}{#2}%
}

\newcommand{\IncludeGraphics}[2][]{%
   \typeout{}%
   \typeout{Graphics: #2}%
   \typeout{\ includegraphics[#1]{#2}}%
   \includegraphics[#1]{#2}
   \typeout{}%
}

\newcommand{\QS}{\Mh{Q}}%
\newcommand{\PS}{\Mh{P}}%
\newcommand{\SSet}{\Mh{\Xi}}%
\newcommand{\RS}{\Mh{R}}%
\definecolor{almostblack}{rgb}{0, 0, 0.3}
\newcommand{\emphw}[1]{{\textcolor{almostblack}{\emph{#1}}}}%

\newcommand{\pc}{\Mh{c}}%
\newcommand{\pb}{\Mh{b}}%

\newcommand{\pp}{\Mh{p}}%
\newcommand{\pq}{\Mh{q}}%
\newcommand{\pf}{\Mh{f}}%

\newcommand{\CH}{\Mh{\mathcal{CH}}}
\newcommand{\CHX}[1]{\CH\pth{#1}}

\newcommand{\Line}{\Mh{\ell}}%

\newcommand{\Term}[1]{\textsf{#1}}
\newcommand{\LP}{\Term{LP}\xspace}%
\newcommand{\VC}{\Term{VC}\xspace}%

\newcommand{\depthTKY}[2]{\Mh{\mathsf{d_{TK}}}\pth{#1}}%
\newcommand{\LS}{\Mh{L}}%
\newcommand{\LSr}{\Mh{L_{\Mapsfrom}}}%
\newcommand{\RSr}{\Mh{R_{\Mapsto}}}%
\newcommand{\rs}{\Mh{\mathcalb{r}}}%
\newcommand{\ls}{\Mh{\mathcalb{l}}}%
\newcommand{\dualX}[1]{#1^{\Mh{\star}}}%
\newcommand{\hp}{\Mh{\mathcalb{h}^{\!\!+}}}%
\newcommand{\hL}{\Mh{\mathcalb{h}}}%
\newcommand{\TS}{\Mh{T}}%
\newcommand{\pt}{\Mh{t}}%
\newcommand{\etal}{\textit{et~al.}\xspace}%

\newcommand{\eps}{{\varepsilon}}%
\newcommand{\Log}{\Mh{\mathcalb{l}}}%
\newcommand{\rank}{\Mh{r}}%

\newcommand{\TLPY}[2]{T_{\LP}\pth{#1, #2}}%
\newcommand{\Ow}{\mathcal{O}_w}
\newcommand{\Soberon}{Sober{\'{o}}n\xspace}%
\newcommand{\Caratheodory}{Carath\'eodory\xspace}
\newcommand{\Tdepth}{\texttt{T}-depth\xspace}%
\newcommand{\Ball}{\Mh{B}}%
\newcommand{\Vercica}{Vre{\'c}ica\xspace}

\newcommand{\seg}{\Mh{s}}%

\newcommand{\hplane}{\Mh{h}}%

\newcommand{\bd}{\partial}%

\BibLatexMode{%
   \bibliography{triangle_cover} }

\title{Improved Approximation Algorithms for Tverberg Partitions}

\SoCGVer{%
   \author{Sariel Har-Peled}%
   {Department of Computer Science, University of Illinois, 201
      N. Goodwin Avenue, Urbana, IL 61801, USA}%
   {sariel@illinois.edu}%
   {https://orcid.org/0000-0003-2638-9635}%
   {Work on this paper was partially supported by a NSF AF award
      CCF-1907400.}%
   \author{Timothy Zhou}%
   {Department of Computer Science, University of Illinois, 201
      N. Goodwin Avenue, Urbana, IL 61801, USA}%
   {}%
   {}%
   {Work on this paper was partially supported by a NSF AF award
      CCF-1907400.}%
}%

\NotSoCGVer{%
   \author{%
      Sariel Har-Peled%
      \SarielThanks{Work on this paper was partially supported by a
         NSF AF award CCF-1907400.  %
      }%
      \and %
      Timothy Zhou%
      \TimThanks{}%
   }%
}

\SoCGVer{%
   \authorrunning{S. Har-Peled and T, Zhou} %
   \Copyright{Sariel Har-Peled and Timothy Zhou} \ccsdesc[500]{Theory
      of computation~Computational geometry} \keywords{Geometric
      spanners, vertex failures, robustness}%
}

\NotSoCGVer{%
   \date{\today}%
}%

\begin{document}

\maketitle

\begin{abstract}
    Tverberg's theorem states that a set of $n$ points in $\Re^d$ can
    be partitioned into $\ceil{n/(d+1)}$ sets whose convex hulls all
    intersect. A point in the intersection (aka Tverberg point) is a
    centerpoint, or high-dimensional median, of the input point set.
    While randomized algorithms exist to find centerpoints with some
    failure probability, a partition for a Tverberg point provides a
    certificate of its correctness.

    Unfortunately, known algorithms for computing exact Tverberg
    points take $n^{O(d^2)}$ time. We provide several new
    approximation algorithms for this problem, which improve running
    time or approximation quality over previous work. In particular,
    we provide the first strongly polynomial (in both $n$ and $d$)
    approximation algorithm for finding a Tverberg point.
\end{abstract}

\section{Introduction}

Given a set $\PS$ of $n$ points in the plane and a query point $\pq$,
classification problems ask whether $\pq$ belongs to the same class as
$\PS$.  Some algorithms use the convex hull $\CHX{\PS}$ as a decision
boundary for classifying $\pq$. However, in realistic datasets, $\PS$
may be noisy and contain outliers, and even one faraway point can
dramatically enlarge the hull of $\PS$.  Thus, we would like to
measure how deeply $\pq$ lies within $\PS$ in way that is more robust
against noise.

In this paper, we investigate the notion of Tverberg depth. However,
there are many related measures of depth in the literature, including:
\begin{compactenumA}
    \smallskip%
    \item \textbf{Tukey depth}. The Tukey depth of $\pq$ is the
    minimum number of points that must be removed before $\pq$ becomes
    a vertex of the convex hull. Computing the depth is equivalent to
    computing the closed halfspace that contains $\pq$ and the
    smallest number of points of $\PS$, and this takes $O(n \log n)$
    time in the plane \cite{c-oramt-04}.

    \smallskip%
    \item \textbf{Centerpoint}. In $\Re^d$, a point with Tukey depth
    $n\alpha$ is an $\alpha$-\emphw{centerpoint}. There is always a
    $1/(d+1)$-centerpoint, known simply as the centerpoint, which can
    be computed exactly in $O(n^{d-1})$ time
    \cite{jm-ccfps-94,c-oramt-04}. It can be approximated using the
    centerpoint of a sample \cite{cemst-acpir-96}, but getting a
    polynomial-time (in both $n$ and $d$) approximation algorithm
    proved challenging.  Clarkson \etal \cite{cemst-acpir-96} provided
    an algorithm that computes a $1/4d^2$-centerpoint in roughly
    $O(d^9)$ time.  Miller and Sheehy \cite{ms-acp-10} derandomized it
    to find a (roughly) $1/2d^2$-centerpoint in $n^{O(\log d)}$
    time. More recently, Har-Peled and Mitchell \cite{hj-jcps-19}
    improved the running time to compute a (roughly)
    $1/d^2$-centerpoint in (roughly) $O(d^7)$ time.

    \smallskip%
    \item \textbf{Onion depth}. Imagine peeling away the vertices of
    the current convex hull and removing them from $\PS$. The onion
    depth is the number of layers which must be removed before the
    point $\pq$ is exposed. The convex layers of points in the plane
    can be computed in $O(n \log n)$ time by an algorithm of Chazelle
    \cite{c-clps-85}. The structure of convex layers is
    well-understood for random points \cite{d-co-04} and grid points
    \cite{hl-pg-13}.

    \smallskip%
    \item \textbf{Uncertainty.} Another model considers uncertainty
    about the locations of the points. Suppose that each point of
    $\PS$ has a certain probability of existing, or alternatively, its
    location is given via a distribution. The depth of query point
    $\pq$ is the probability that $\pq$ is in the convex hull once
    $\PS$ has been sampled.  Under certain assumptions, this
    probability can be computed exactly in $O(n \log n)$ time
    \cite{ahsyz-chuu-17}. Unfortunately, the computed value might be
    very close to zero or one, and therefore tricky to interpret.

    \smallskip%
    \item \textbf{Simplicial depth.}  The simplicial depth of $\pq$ is
    the number of simplices induced by $\PS$ containing it.  This
    number can be approximated quickly after some preprocessing
    \cite{ass-asd-15}. However, it can be quite large for a point
    which is intuitively shallow.
\end{compactenumA}

\paragraph*{Tverberg depth.}
Given a set $\PS$ of $n$ points in $\Re^d$, a \emphi{Tverberg
   partition} is a partition of $\PS$ into $k$ disjoint sets
$\PS_1, \ldots, \PS_k$ such that $\bigcap_i \CHX{\PS_i}$ is not
empty. A point in this intersection is a \emphi{Tverberg
   point}. Tverberg's theorem states that $\PS$ has a Tverberg
partition into $\ceil{n/(d+1)}$ sets.  In particular, the
\emph{Tverberg depth} (\emphw{\Tdepth}) of a point $\pq$ is the
maximum size $k$ of a Tverberg partition such that
$\pq \in \bigcap_{i=1}^k \CHX{\PS_i}$.

By definition, points of \Tdepth $n/(d+1)$ are centerpoints for $\PS$.
In the plane, Reay \cite{r-sgtt-79} showed that if a point has Tukey
depth $k \leq \cardin{\PS}/3$, then the \Tdepth of $\pq$ is $k$. This
property is already false in three dimensions \cite{a-mcpcm-93}.  The
two-dimensional case was handled by Birch \cite{b-o3pp-59}, who proved
that any set of $n$ points in the plane can be partitioned into $n/3$
triples whose induced triangles have a common intersection point.

\paragraph*{Computing a Tverberg point.}
For work on computing approximate Tverberg points, see
\cite{ms-acp-10,mw-atplt-13,rs-aatt-16,cm-ndtta-20} and the references
therein.  Currently, no polynomial-time (in both $n$ and $d$)
approximation algorithm is known for computing Tverberg points.  This
search problem is believed to be quite hard, see \cite{mmss-relpf-17}.

Algorithms for computing an \emph{exact} Tverberg point of \Tdepth
$n/(d+1)$ implement the construction implied by the original
proof. The runtime of such an algorithm is
\begin{math}
    d^{O(d^2)}n^{d(d+1) + 1},
\end{math}
see \secref{tverberg:exact}.  As previously mentioned, the exception
is in two dimensions, where the algorithm of Birch \cite{b-o3pp-59}
runs in $O(n\log n)$ time. But even in three dimensions, we are
unaware of an algorithm faster than $O(n^{13})$.

\paragraph*{Convex combinations and \Caratheodory's theorem.}

The challenge in finding a Tverberg point is that we have few
subroutines at our disposal with runtimes polynomial in $d$. Consider
the most basic task -- given a set $\PS$ of $n$ points and a query
point $\pq$, decide if $\pq$ lies inside $\CHX{\PS}$, and if so,
compute the convex combination of $\pq$ in term of the points of
$\PS$. This problem can be reduced to linear programming. Currently,
the fastest strongly polynomial \LP algorithms run in super-polynomial
time $2^{O(\sqrt{d \log d})} + O(d^2 n)$
\cite{c-lvali-95,msw-sblp-96}, where $d$ is the number of variables
and $n$ is the number of constraints.  However, any given convex
combination of $\PS$ representing $\pq$ can be sparsified in
polynomial time into a convex combination using only $d+1$ points of
$\PS$.  \lemref{c:lp} describes this algorithmic version of
\Caratheodory's theorem.

\paragraph*{Radon partitions in polynomial time.}
Finding points of \Tdepth $2$ is relatively easy.  Any set of $d+2$
points in $\Re^d$ can be partitioned into two disjoint sets whose
convex hulls intersect, and a point in the intersection is a
\emphw{Radon point}. Radon points can be computed in $O(d^3)$ time by
solving a linear system with $d+2$ variables.  Almost all the
algorithms for finding Tverberg points mentioned above amplify the
algorithm for finding Radon points.

\newcommand{\MPX}[1]{%
   \begin{minipage}{\SoCGVer{6cm}\NotSoCGVer{9cm}}
       #1
   \end{minipage}%
} \newcommand{\MPFX}[1]{%
   \begin{minipage}{2cm}
       \smallskip%
       #1 \smallskip%
   \end{minipage}%
}

\begin{figure}[p]
    \centering%
    \begin{tabular}{|c|c|l|}
      \hline
      Depth%
      &%
        Running time
      &%
        Ref / Comment
      \\
      \hline%
      $n/(d+1)$
      & %
        $d^{O(d^2)}n^{d(d+1) +1}\Bigr.$
      & %
        Tverberg theorem
      \\
      \hline%
      $d=2:\, n/3$
      & %
        $O(n \log n)$
      & %
        \MPX{
        \cite{b-o3pp-59}:
        \thmref{birch}%
        }\\
      \hline%
      \MPFX{%
      $\ds \frac{n}{2(d+1)^2}\Bigr.$%
      }
      & %
        $n^{O(\log d )}$
      & %
        \MPX{Miller and Sheehy \cite{ms-acp-10}}
      \\
      \hline%
      \MPFX{
      $\ds \frac{n}{4(d+1)^3}\Bigr.$}
      & %
        $d^{O(\log d )}n$
      & %
        \MPX{Mulzer and Werner
        \cite{mw-atplt-13}}
      \\
      \hline
      \MPFX{$\ds \frac{n}{2d(d+1)^2}\Bigr.$}%
      & %
        $\Ow(n^4)$
      & %
        \MPX{Rolnick and \Soberon \cite{rs-aatt-16}}%
      \\
      \hline
      \MPFX{$\ds \frac{(1-\delta)n}{d(d+1)}\Bigr.$}%
      & %
        $(d/\delta)^{O(d)} + \Ow(n^4)$
      & %
        \MPX{Rolnick and \Soberon \cite{rs-aatt-16}}%
      \\
      \hline%
      \multicolumn{3}{c}{$\Bigl.$New results}
      \\
      \hline%
      $\ds \frac{(1-\delta)n}{2(d+1)^2}\Bigr.$
      & %
        $d^{O(\log (d/\delta))}n$
      & %
        \thmref{m:w:better}%
      \\
      \hline%
      \begin{minipage}{2cm}
          \smallskip%
          $\ds \frac{n^{~}}{O(d^2 \log d)}\Bigr.$
      \end{minipage}
      &
        $O(dn)$
      & %
        \MPX{\lemref{tverberg:1}:
        Only partition}
      \\
      \hline%
      \begin{minipage}{2cm}
          \smallskip%
          $\ds \frac{n^{~}}{O(d^3 \log d)}\Bigr.$ \smallskip%
      \end{minipage}
      &
        $O(dn + d^{7} \log^{6} d)$
      & %
        \MPX{%
        \lemref{partition}:
        Partition + point, but
        no convex combination}
      \\
      \hline%
      \MPFX{$\ds \frac{n^{~}}{O(d^2 \log d)}\Bigr.$}
      & %
        $\Ow(n^{5/2} + nd^3)$
      & %
        \MPX{\lemref{tverberg:1:lp}:
        Weakly polynomial}
      \\
      \hline%
      \MPFX{
      $\ds\frac{(1-\delta)n}{2(d+1)^2}$%
      }
      & %
        $d^{O(\log \log (d/\delta))}\Ow( n^{5/2}) $
      & %
        \MPX{
        \thmref{m:w:better:lp}:
        Weakly quasi
        polynomial}
      \\
      \hline%
      \MPFX{%
      $\ds\frac{(1-\delta)n}{d(d+1)}$}
      & %
        \begin{minipage}{4cm}
            \smallskip%
            \begin{math}
                O(c + c' n + d^2 n \log^2 n )
            \end{math}\\
            $c= d^{O(d)} /\delta^{2(d-1)}$\\
            $c'= 2^{O(\sqrt{d\log d})}$ \smallskip%
        \end{minipage}
      & %
        \MPX{
        \lemref{t:low:dim}:
        Useful for low dimensions
        }
      \\
      \hline
    \end{tabular}
    \caption{The known and improved results for Tverberg partition.
       The notation $\Ow$ hides terms with polylogarithmic dependency
       on size of the numbers, see \remref{width}. The parameter
       $\delta$ can be freely chosen. }
    \figlab{results}
\end{figure}

\begin{figure}
    \centering%
    \begin{tabular}{|c|c||{c}||c|c|l|}
      \hline
      Dim
      & T. Depth
      & New depth
      & Known
      & Ref
      & Comment\\
      \hline%
      $3$
      & $n/4$
      & $n/6$
      & $n/8$
      &
      &
      \\
      $4$
      & $n/5$
      & $n/9$
      & $n/16$
      &
        \cite{mw-atplt-13}
      &
      \\
      $5$
      & $n/6$
      & $n/18$
      & $n/32$
      &
      &
      \\
      \hline
      \hline
      $6$
      & $n/7$
      & $n/27$
      & $(1-\delta)n/42$
      &
      & Original paper describes  a weakly
      \\
      $7$
      & $n/8$
      & $n/54$
      & $(1-\delta)n/56$
      &
        \cite{rs-aatt-16}
      & %
        polynomial algorithm. The improved
      \\
      $8$
      & $n/9$
      & $n/81$
      & $(1-\delta)n/72$
      &
      &%
        algorithm is described in \lemref{t:low:dim}.%
      \\
      \hline
    \end{tabular}%
    \caption{The best approximation ratios for Tverberg depth in low
       dimensions, with nearly linear-time algorithms, as implied by
       \lemref{t:4:d}. Note that the new algorithm is no longer an
       improvement in dimension $8$. We are unaware of any better
       approximation algorithms (except for running the exact
       algorithm for Tverberg's point, which requires $n^{O(d^2)}$
       time).}
    \figlab{low:dim}
\end{figure}

\paragraph*{Our results.}
The known and new results are summarized in \figref{results}. In
\secref{prelims}, we review preliminary information and known results,
which include the following.
\begin{compactenumI}[leftmargin=0.9cm]
    \smallskip%
    \item \textbf{An exact algorithm.}  The proof of Tverberg's
    theorem is constructive and leads to an algorithm with running
    time $O(n^{d(d+1)+1})$. It seems that the algorithm has not been
    described and analyzed explicitly in the literature. For the sake
    of completeness, we provide this analysis in
    \secref{tverberg:exact}.

    \smallskip%
    \item \textbf{In two dimensions.}  Given a set $\PS$ of $n$ points
    in the plane and a query point $\pq$ of Tukey depth $k$, Birch's
    theorem \cite{b-o3pp-59} implies that $\pq$ can be covered by
    $\min( k, \floor{n/3})$ vertex-disjoint triangles of $\PS$.  One
    can compute $k$ and this triangle cover in $O(n + k\log k)$ time,
    and use them to compute a Tverberg point of depth $\floor{n/3}$ in
    $O(n \log n)$ time. For the sake completeness, this is described
    in \secref{2:d}.
\end{compactenumI}
\medskip%
\noindent
In \secref{improved}, we provide improved algorithms for computing
Tverberg points and partitions.
\begin{compactenumI}[leftmargin=0.9cm]
    \smallskip%
    \item \textbf{Projections in low dimensions.}  We use projections
    to find improved approximation algorithms in dimensions $3$ to
    $7$, see \figref{low:dim}.  For example, in three dimensions, one
    can compute a point with \Tdepth $n/6$ in $O(n \log n)$ time.

    \smallskip%
    \item \textbf{An improved quasi-polynomial algorithm.}  We modify
    the algorithm of Miller and Sheehy to use a buffer of free
    points. Coupled with the algorithm of Mulzer and Werner
    \cite{mw-atplt-13}, this idea yields an algorithm that computes a
    point of \Tdepth $\geq (1-\delta)n/2(d+1)^2$ in
    $d^{O(\log( d/\delta))}n$ time. This improves the approximation
    quality of the algorithm of \cite{mw-atplt-13} by a factor of
    $2(d+1)$, while keeping (essentially) the same running time.

    \smallskip%
    \item \textbf{A strongly polynomial algorithm.}  In
    \secref{strong}, we present the first strongly polynomial
    approximation algorithm for Tverberg points, with the following
    caveats:
    \begin{compactenumi}
        \item the algorithm is randomized, and might fail,

        \smallskip%
        \item one version returns a Tverberg partition, but not a
        point that lies in its intersection,

        \smallskip%
        \item the other (inferior) version returns a Tverberg point
        and a partition realizing it, but not the convex combination
        of the Tverberg point for each set in the partition.
    \end{compactenumi}

    Specifically, one can compute a partition of $\PS$ into
    $n/O(d^2 \log d)$ sets, such that the intersection of their convex
    hulls is nonempty (with probability close to one), but without
    finding a point in the intersection.  Alternatively, one can also
    compute a Tverberg point, but the number of sets in the partition
    decreases to $n/O(d^3 \log d)$.

    \smallskip%
    \item \textbf{A weakly polynomial algorithm.} %
    Revisiting an idea of Rolnick and \Soberon \cite{rs-aatt-16}, we
    use algorithms for solving \LP{}s. The resulting running time is
    either weakly polynomial (depending logarithmically on the
    relative sizes of the numbers in the input) or super-polynomial,
    depending on the \LP solver. In particular, the randomized,
    strongly polynomial algorithms described above can be converted
    into constructive algorithms that compute the convex combination
    of the Tverberg point over each set in its partition.  Having
    computed approximate Tverberg points of \Tdepth
    $\geq n/O(d^2 \log d)$, we can feed them into the buffered version
    of Miller and Sheehy's algorithm to compute Tverberg points of
    depth $\geq (1-\delta)n/2(d+1)^2$. This takes
    $d^{O( \log \log (d/\delta))} \Ow(n^{5/2})$ time, where $\Ow$
    hides polylogarithmic terms in the size of the numbers involved,
    see \remref{width}.

    \smallskip%
    \item \textbf{Faster approximation in low dimensions.} %
    One can compute (or approximate) a centerpoint, then repeatedly
    extract simplices covering it until the centerpoint is
    exposed. This leads to an $O_d(n^2)$ approximation algorithm
    \cite{rs-aatt-16}. Since $O_d$ hides constants that depend badly
    on $d$, this method is most useful in low dimensions. By random
    sampling, we can speed up this algorithm to $O_d(n \log^2 n)$
    time.
\end{compactenumI}

\section{Background, preliminaries and known results}
\seclab{prelims}

In this section, we cover known results in the literature about
Tverberg partition. We provide proofs for many of the claims or the
sake of completeness, as we use them later in the paper.

\subsection{Definitions}
\begin{definition}
    \deflab{log}%
    A \emphi{Tverberg partition} (or a \emphi{log}) of a set of points
    $\PS \subseteq \Re^d$, for a point $\pq$, is a set
    $\Log = \{\PS_1, \ldots, \PS_k\} $ of vertex-disjoint subsets of
    $\PS$, each containing at most $d+1$ points, such that
    $\pq \in \CHX{\PS_i}$, for all $i$.  The \emphi{rank} of $\Log$ is
    $k=|\Log|$. The maximum rank of any log of $\pq$ is the
    \emphi{Tverberg depth} (or \emphi{\Tdepth}) of $\pq$.

    A set $\PS_j$ in a log is a \emphi{batch}.  For every batch
    $\PS_j =\{\pp_1, \ldots, \pp_{d+1}\}$ in the log, we also store
    the convex coefficients $\alpha_1, \ldots, \alpha_{d+1} \geq 0$
    such that $\sum_i \alpha_i \pp_i = \pq$ and $\sum_i \alpha_i =
    1$. A pair $(\pq, \Log)$ of a point and its log is a \emphi{site}.
\end{definition}

Tverberg's theorem states that, for any set of $n$ points in $\Re^d$,
there is a point in $\Re^d$ with Tverberg depth $\ceil{n/(d+1)}$.  For
simplicity, we assume the input is in general position.

\subsection{An exact algorithm}
\seclab{tverberg:exact}

The constructive proof of Tverberg and \Vercica \cite{tv-grths-93}
implies an algorithm for computing exact Tverberg points. We include
the proof for the sake of completeness, as we also provide an analysis
of the running time.

\begin{lemma}[\cite{tv-grths-93}]
    \lemlab{tverberg:exact:alg}%
    Let $\PS$ be a set of $n$ points in $\Re^d$. In
    \begin{math}
        d^{O(d^2)}n^{d(d+1) + 1}
    \end{math}
    time, one can compute a point $\pq$ and a partition of $\PS$ into
    disjoint sets $\PS_1, \ldots, \PS_r$ such that
    $\pq \in \bigcap_i \CHX{\PS_i}$, where $r=\ceil{ n/(d+1)}$.
\end{lemma}

\begin{proof}
    The constructive proof works by showing that a local search
    through the space of partitions stops when arriving at the desired
    partition.  To simplify the exposition, we assume that $\PS$ is in
    general position, and $n = r (d+1)$.  The algorithm starts with an
    arbitrary partition of $\PS$ into $r$ sets $\PS_1, \ldots, \PS_r$,
    all of size $d+1$.

    \begin{figure}[h]
        \centering%
        \includegraphics{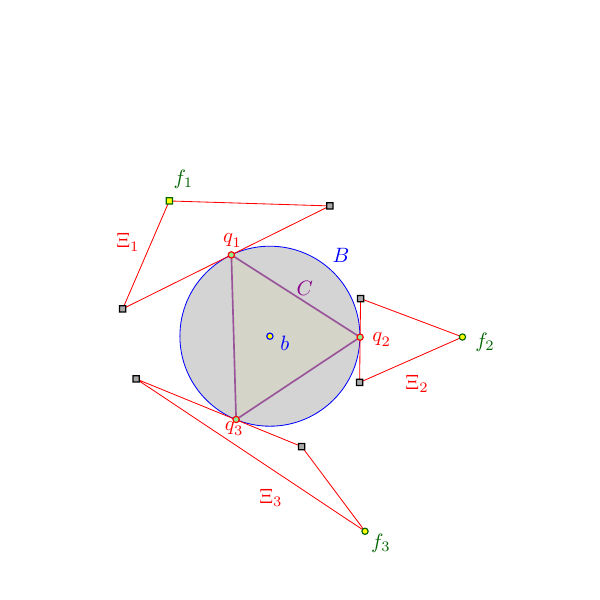}
        \caption{The configuration before the exchange}
        \figlab{before:exchange}
    \end{figure}

    In each iteration, the algorithm computes the ball $\Ball$ of
    minimum radius that intersects all the convex hulls
    $\CHX{\PS_1}, \ldots, \CHX{\PS_r}$.  Since $\PS$ is in general
    position, $\Ball$ is tangent to the hulls of at most $k \leq d+1$
    sets of the partition, say $\SSet_1, \ldots, \SSet_k$.  The $k$
    sets defining this ball are \emphw{tight}. Let
    $\pq_i = \Ball \cap \CHX{\SSet_i}$, for $i=1,\ldots, k$, and let
    $\pb$ be the center of $\Ball$, see
    \figref{before:exchange}. Observe that
    $\pb \in C = \CHX{ \{ \pq_1, \ldots, \pq_k\}}$ as otherwise one
    can decrease the radius of $\Ball$ by moving its center towards
    $C$ (which contradicts the minimality of $\Ball$).

    For $i=1,\ldots, k$, the point $\pq_i$ lies on the boundary of the
    simplex $\CHX{\SSet_i}$. As such, there exists a ``free'' point
    $\pf_i \in \SSet_i$, such that $\pq_i \in \CHX{\SSet_i -
       \pf_i}$. Let $\hplane$ be the hyperplane passing through $\pb$
    which is orthogonal to $\pf_1 - \pb$. Let $\hplane^+$ be the open
    halfspace bounded by $\hplane$ that does not contain $\pf_1$.  By
    our general position assumption, no point of
    $\QS = \{\pq_1, \ldots, \pq_k\}$ lies on $\hplane$ (to see that,
    consider perturbing the points randomly, and observe that the
    probability for this event to happen is zero). As such, if
    $\hplane^+ \cap \QS$ is empty, then $\QS$ is contained in an open
    hemisphere $\bd \Ball \cap \hplane^- $, where $\hplane^-$ is the
    other open halfspace bounded by $\hplane$. But then, $\Ball$ can
    be shrunk further. As such, we conclude that $\QS \cap \hplane^+$
    is not empty. Suppose that $\pq_2$ is in this intersection, see
    \figref{exchange}.
    \begin{figure}
        \includegraphics[page=2]{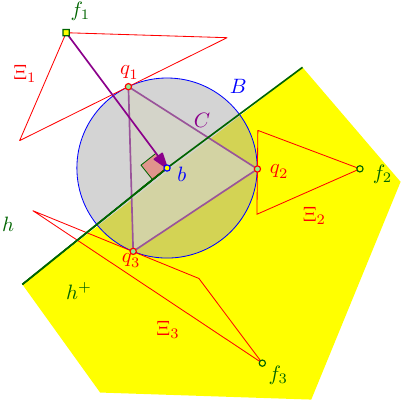} \hfill%
        \includegraphics[page=3]{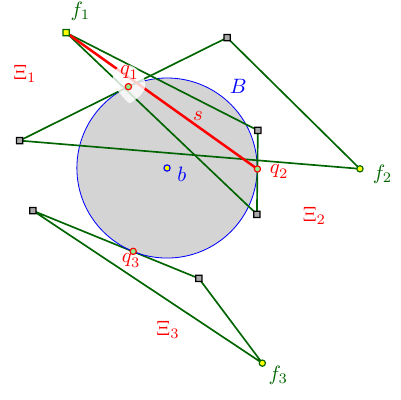}
        \caption{A beneficial exchange.}
        \figlab{exchange}
    \end{figure}

    In particular, consider the exchange of $\pf_1$ and $\pf_2$:
    \begin{equation*}
        \SSet_1' = \SSet_1 - \pf_1 + \pf_2
        \qquad\text{and}\qquad%
        \SSet_2' = \SSet_2 - \pf_2 + \pf_1.
    \end{equation*}
    The segment $\seg = \pf_1 \pq_2 \subseteq \CHX{\SSet_2'}$
    intersects the interior of $\Ball$, so $\Ball$ can be shrunk after
    the exchange.

    In each iteration, the radius of the ball strictly decreases.  As
    such, the algorithm can find exchanges which allow $\Ball$ to
    shrink, as long as the radius of $\Ball$ is larger than zero. When
    it stops, the algorithm will have computed a partition of $P$ into
    sets whose convex hulls all intersect.

    The ball maintained by the algorithm is uniquely defined by the
    $d+1$ tight sets, where each tight set contains exactly $d+1$
    points. As such, the number of possible balls considered by the
    algorithm, and thus the number of iterations, is bounded by
    $O( n^{(d+1)^2})$. One can get a slightly better bound, by
    observing that the free point in each tight set is irrelevant in
    defining the smallest ball. As such, the number of different balls
    that might be computed by the algorithm is bounded by
    $O(n^{d(d+1)})$.

    \smallskip%
    We next provide some low-level details. As a starting point, one
    needs the following two geometric primitives:
    \begin{compactenumI}
        \smallskip%
        \item Given a point $\pp$ (or a ball), and a simplex in
        $\Re^d$, compute the distance from the point to the simplex. A
        brute force algorithm for this works in $2^{d} d^{O(1)}$ time,
        as deciding if a point is in a simplex can be done in $O(d^4)$
        time by computing $d+1$ determinants. If $\pp$ is not in the
        interior of the simplex, then we recurse on each of the $d+1$
        facets of the simplex, projecting $\pp$ to each subspace
        spanning this subset, and compute the nearest-point problem
        recursively, returning the best candidate returned.  Since
        there are $2^{d+1}$ subsimplices, the running time stated
        follows.

        \smallskip%
        \item Given $d+1$ simplices each defined by $d+1$ points of
        $\PS \subseteq \Re^d$, the task is to compute the minimum
        radius ball which intersects all of them in $O(1)$ time.  It
        can be easily written as a convex program in $O(d^2)$
        variables, which such be solved exactly in $d^{O(d^2)}$
        time. An alternative way to get a similar running time is to
        explicitly write down the distance function induced by each
        simplex, then compute the lowest point in the upper envelope
        of the resulting set of functions. There are $D \leq d^{O(d)}$
        functions, so the complexity of the arrangement of their
        images is in $\Re^{d+1}$ is $O(D^{d+1})$.  This arrangement,
        and thus the lower point in the upper envelope, can be
        computed in the time stated.
    \end{compactenumI}

    \medskip%
    \noindent%
    At the beginning of each iteration, the algorithm computes the
    smallest ball intersecting the convex hulls of the sets in the
    \emph{current} partition as follows: This is an \LP type problem,
    and using the above geometric primitives, one can deploy any
    linear-time algorithm for this problem \cite{h-gaa-11,
       c-lvali-95}. The resulting running time is $d^{O(d^2)} n$.
    This also returns (i) the tight sets, (ii) the ball $\Ball$ and
    (iii) the points $\pq_1, \ldots, \pq_{d+1}$. It is now
    straightforward to compute the free point in each simplex and do
    the exchange in polynomial time in $d$. Putting everything
    together, we get running time
    \begin{math}
        n^{d(d+1)} d^{O(d^2)} n = d^{O(d^2)}n^{d(d+1) + 1} .
    \end{math}
\end{proof}

\subsection{In two dimensions}
\seclab{2:d}%

\subsubsection{Computing Tukey depth in the plane}

We first review Tukey depth, which is closely related to Tverberg
depth in two dimensions.
\begin{defn}
    The \emphi{Tukey depth} of a point $\pq$ in a set
    $\PS \subseteq \Re^d$, denoted by $\depthTKY{\pq}{\PS}$, is the
    minimum number of points contained in any closed halfspace
    containing $\pq$.
\end{defn}

The following result is implicit in the work of Chan \cite[Theorem
5.2]{c-garot-99}.  Chan solves the decision version of the Tukey depth
problem, while we need to compute it explicitly, resulting in a more
involved algorithm.

\begin{lemma}
    \lemlab{tukey:depth:2:d}%
    Given a set $\PS$ of $n$ points in the plane and a query point
    $\pq$, such that $\PS \cup \{ \pq\}$ is in general position, one
    can compute, in $O(n + k \log k)$ time, the Tukey depth $k$ of
    $\pq$ in $\PS$. The algorithm also computes the halfplane
    realizing this depth.%
\end{lemma}

\begin{proof}
    In the dual, $\dualX{\pq}$ is a line, and the task is to find a
    point on this line which minimizes the number of lines of
    $\dualX{\PS}$ (i.e., set of lines dual to the points of $\PS$)
    strictly below it.  More precisely, one has to also solve the
    upward version, and return the minimum of the two
    solutions. Handling the downward version first, every point
    $\pp \in \PS$ has a dual line $\dualX{\pp}$. The portion of
    $\dualX{\pq}$ that lies above $\dualX{\pp}$ is a closed ray on
    $\dualX{\pq}$.  As such, we have a set of rays on the line (which
    can be interpreted as the $x$-axis), and the task is to find a
    point on the line contained in the minimum number of rays. (This
    is known as linear programming with violations in one dimension.)

    Let $\RSr$ (resp. $\LSr$) be the set of points that corresponds to
    heads of rays pointing to the right (resp. to the left) by
    $\dualX{\PS}$ on $\dualX{\pq}$. Let $\RS_i$ be the set of
    $\floor{n/2^i}$ rightmost points of $\RS_{\Mapsto}$, for
    $i=0,\ldots, h = \floor{ \log_2 n}$. Using median selection, each
    set $\RS_i$ can be computed from $\RS_{i-1}$ in
    $\cardin{\RS_{i-1}}$ time. As such, all these sets can be computed
    in $\sum_i O(n/2^i) = O( n)$ time. The sets
    $\LS_0, \LS_1, \ldots, \LS_h$ are computed in a similar fashion.
    For all $i$, we also compute the rightmost point $\rs_{i-1}$ of
    $\RS_{i-1} \setminus \RS_i$ (which is the rightmost point in
    $\RS \setminus \RS_i$). Similarly, $\ls_i$ is the leftmost point
    of $\LS_{i-1} \setminus \LS_i$, for $i=1, \ldots, h$. Let $\rs_0$
    (resp. $\ls_0$) be the rightmost (resp. leftmost) point of $\RS_0$
    (resp. $\LS_0$).

    Now, compute the maximum $j$ such that $\rs_j$ is to the left of
    $\ls_j$. Observe that
    $(\RS \setminus \RS_j) \cup (\LS \setminus \LS_j)$ form a set of
    rays that their intersection is non-empty (i.e.,
    feasible). Similarly, the set of rays
    $(\RS \setminus \RS_{j+1}) \cup (\LS \setminus \LS_{j+1})$ is not
    feasible, and any set of feasible rays must be created by removing
    at least
    $\cardin{\RS_{j+1}} = \cardin{\LS_{j+1}} = \floor{n/2^{j+1}}$ rays
    from $\RSr \cup \LSr$. Hence, in linear time, we have computed a
    $4$-approximation to the minimum number of rays that must be
    removed for feasibility. In particular, if $S$ is a set of rays
    such that $(\RSr \cup \LSr) \setminus S$ is feasible, then
    $(\RS_{j-1} \cup \LS_{j-1}) \setminus S$ is also feasible.
    Namely, if the minimum size of such a removeable set $S$ is $k$,
    then we have computed a set of $O(k)$ rays, such that it suffices
    to solve the problem on this smaller set.

    In the second stage, we solve the problem on
    $\RS_{j-1} \cup \LS_{j-1}$. We first sort the points, and then for
    each point in this set, we compute how many rays must be removed
    before it lies in the intersection of the remaining rays.  Given a
    location $\pp$ on the line, we need to remove all the rays of
    $\RS_{j-1}$ (resp.  $\LS_{j-1}$) whose heads lie to the right
    (resp. left) of $\pp$. This can be done in $O( k \log k)$ time by
    sweeping from left to right and keeping track of the rays that
    need to be removed.

    As such, we can solve the \LP with violations on the line in
    $O( n + k \log k)$ time, where $k$ is the minimum number of
    violated constraints. A $4$-approximation to $k$ can be computed
    in $O(n)$ time.

    Now we return to the Tukey depth problem. First we compute a
    $4$-approximation, denoted by $\widetilde{k_{\downarrow}}$, for
    the minimum number of lines crossed by a vertical ray shot down
    from a point on $\dualX{\pq}$. Similarly, we compute
    $\widetilde{k_{\uparrow}}$. If
    $4\widetilde{k_\downarrow} < \widetilde{k_\uparrow}$, then we
    compute $k_\downarrow$ exactly and return the point on
    $\dualX{\pq}$ that realizes it. (In the primal, this corresponds
    to a closed halfspace containing $\pq$ along with exactly
    $k_\downarrow$ points of $\PS$.)  Similarly, if
    $\widetilde{k_\downarrow} > 4\widetilde{k_\uparrow}$, then we
    compute $k_\uparrow$ exactly, and return it as the desired
    solution. In the remaining case, we compute both quantities and
    return the minimum of the two.
\end{proof}

\subsubsection{Computing a log realizing the Tukey depth of a point}

For shallow points in the plane, the Tukey and Tverberg depths are
equivalent, and we can compute the associated Tverberg partition.
This is essentially implied by the work of Birch \cite{b-o3pp-59}, and
we include the details for the sake of completeness.

\begin{lemma}
    \lemlab{r:shallow}%
    Let $\PS$ be a set of $n$ points in the plane, let $\pq$ be a
    query point such that $\PS \cup \{ \pq \}$ is in general position,
    and suppose that $k = \depthTKY{\pq}{\PS} \leq n/3$. Then one can
    compute a log for $q$ of rank $k$ in $O(n + k \log k)$ time.
\end{lemma}
\begin{proof}
    Using the algorithm of \lemref{tukey:depth:2:d}, compute the Tukey
    depth of $\pq$ and the closed halfplane $\hp$ realizing it, where
    $\pq$ lies on the line $\hL$ bounding $\hp$. This takes
    $O(n + k \log k)$ time.  By translation and rotation, we can
    assume that $\pq$ is the origin and $\hL$ is the $x$-axis. Let
    $\PS^+ = \PS \cap \hp$ be the $k$ points realizing the Tukey depth
    of $\pq$, and let $\PS^- = \PS \setminus \PS^+$ be the set of
    points below the $x$-axis, so that $\cardin{\PS^-} = n-k \geq 2k$.

    Consider the counterclockwise order of the points of $\PS^-$
    starting from the negative side of the $x$-axis. Let $\TS$ be the
    set containing the first and last $k$ points in this order,
    computed in $O(n)$ time by performing median selection twice. Let
    $\TS = \{ \pt_1, \ldots, \pt_{2k} \}$ be the points of $\TS$
    sorted in counterclockwise order. Similarly, let
    $\PS^+ = \{ \pp_1, \ldots, \pp_k\}$ be the points of $\PS^+$
    sorted in counterclockwise order starting from the positive side
    of the $x$-axis, see \figref{r:shallow}.

    \begin{figure}[b]
        \begin{minipage}{0.9999\linewidth}
            \begin{minipage}{0.3\linewidth}

                \bigskip

                \captionof{figure}{Illustration of the proof of
                   \lemref{r:shallow}.}  %
                \figlab{r:shallow}

            \end{minipage}
            \hfill
            \begin{minipage}{0.34\linewidth}
                \IncludeGraphics[width=0.95\linewidth]{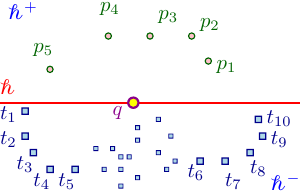}
            \end{minipage}
            \hfill
            \begin{minipage}{0.34\linewidth}
                \IncludeGraphics[page=2,width=0.95\linewidth]{figs/triangles}
            \end{minipage}

        \end{minipage}
    \end{figure}

    Let $\triangle_i = \triangle \pp_i \pt_i \pt_{k+i}$, for
    $i=1,\ldots, k$. We claim that $\pq \in \triangle_i$ for all
    $i$. To this end, let $\Line_i$ be the line passing through the
    origin and $\pp_i$, and let $\Line_i^+$ denote the halfspace it
    induces to the left of the vector $\overrightarrow{\pq
       \pp_i}$. Then $\Line_i^+$ must contain $\pt_i$, as otherwise,
    $\cardin{\PS \cap \Line_i^+} < k$, contradicting the Tukey depth
    of $\pq$. Namely, the segment $\pp_i \pt_i$ intersects the
    negative side of the $x$-axis. A symmetric argument, applied to
    the complement halfplane, implies that the segment
    $\pt_{k+i} \pp_i$ intersects the positive side of the $x$-axis.
    Then the origin $\pq$ is contained in $\triangle_i$, as claimed.

    Computing these triangles, we have found a log for $\pq$ of rank
    $k$ in $O(n + k \log k)$ time.
\end{proof}

The Tukey depth of a point can be as large as $\floor{n/2}$. Indeed,
consider the vertices of a regular $n$-gon (for odd $n$), the polygon
center has depth $\floor{n/2}$.

\begin{lemma}
    \lemlab{r:deep}%
    Let $\PS$ be a set of $n$ points in the plane, and let $\pq$ be a
    query point such that $\PS \cup \{ \pq \}$ is in general position,
    and $\pq$ has Tukey depth larger than $n/3$. Then, one can
    compute, in $O(n \log n)$ time, a log for $\pq$ of rank
    $s = \floor{n/3}$.
\end{lemma}
\begin{proof}
    Assume that $\pq$ is in the origin, and sort the points of $\PS$
    in counterclockwise order, where $\pp_i$ is the $i$\th point in
    this order. For $i=1,\ldots, s$, let
    $\triangle_i = \triangle \pp^{}_i \pp^{}_{s+i } \pp^{}_{2s+i}$.
    We claim that $\angle \pp_i \pq \pp_{s+i } \leq \pi$. Otherwise,
    there is a halfspace containing fewer than $s$ points induced by
    the line passing through $\pp^{}_i$ and $\pq$. Similarly,
    $\angle \pp^{}_{s+i} \pq \pp^{}_{2s+i } \leq \pi$ and
    $\angle \pp^{}_{2s+i} \pq \pp^{}_{i } \leq \pi$, so that $\pq$
    lies inside $\triangle_i$, as desired.
\end{proof}

\begin{theorem}
    \thmlab{c:support}%
    Let $\PS$ be a set of $n$ points in the plane, let $\pq$ be a
    query point such that $\PS \cup \{ \pq \}$ is in general position,
    and suppose that $k=\depthTKY{\pq}{\PS}$ is the Tukey depth of
    $\pq$ in $\PS$. Then one can compute the Tukey depth $k$ of $\pq$,
    along with a log of rank $\tau = \min( \floor{n/3}, k )$, in
    $O(n + k \log k)$ time.
\end{theorem}
\begin{proof}
    Compute the Tukey depth of $\pq$ in $O(n + k \log k)$ time, using
    the algorithm of \lemref{tukey:depth:2:d}. If $k \leq n/3$, then
    compute the log using the algorithm of \lemref{r:shallow}, and
    otherwise using the algorithm of \lemref{r:deep}.
\end{proof}

The above implies the following theorem of Birch, which predates
Tverberg's theorem.

\begin{theorem}[\cite{b-o3pp-59}]
    \thmlab{birch}%
    Let $\PS$ be a set of $n =3k$ points in the plane. Then there
    exists a partition of $\PS$ into $k$ vertex-disjoint triangles,
    such that their intersection is not empty.  The partition can be
    computed in $O(n \log n)$ time.
\end{theorem}
\begin{proof}
    A centerpoint of $\PS$ can be computed in $O(n \log n)$ time
    \cite{c-oramt-04}. Such a centerpoint has Tukey depth at least
    $k$, so the algorithm of \lemref{r:deep} partitions $\PS$ into $k$
    triangles.
\end{proof}

\begin{remark}
    (A) Note that Tverberg's theorem in the plane is slightly stronger
    -- it states that any point set with $3k-2$ points has Tverberg
    depth $k$.  In such a decomposition, some of the sets may be pairs
    of points or singletons.

    (B) In the colored version of Tverberg's theorem in the plane, one
    is given $3n$ points partitioned into three classes of equal
    size. Agarwal \etal \cite{asw-actp-08} showed how to compute a
    decomposition into $n$ triangles covering a query point, where
    every triangle contains a vertex of each color (if such a
    decomposition exists). This problem is significantly more
    difficult, and their running time is a prohibitive $O(n^{11})$.
\end{remark}

\subsection{Miller and Sheehy's algorithm}

Here, we review Miller and Sheehy's approximation algorithm
\cite{ms-acp-10} for computing a Tverberg point before describing our
improvement.

\paragraph*{Radon partitions.}
Radon's theorem states that a set $\PS$ of $d+2$ points in $\Re^d$ can
be partitioned into two disjoint subsets $\PS_1,\PS_2$ such that
$\CHX{\PS_1} \cap \CHX{\PS_2} \neq \emptyset$. This partition can be
computed via solving a linear system in $d+2$ variables in $O(d^3)$
time. A point in this intersection (which is an immediate byproduct of
computing the partition) is a \emphi{Radon point}.

\bigskip%

Let $\pp$ be a point in $\Re^d$. It is a \emphi{convex combination} of
points $\{\pp_1, \ldots, \pp_m\}$ if there are
$\alpha_1,\ldots, \alpha_m \in [0,1]$ such that
$\pp = \sum_{i=1}^m \alpha_i \pp_i$ and $\sum_i \alpha_i = 1$.

\begin{lemma}[{\normalfont \cite{ms-acp-10}:
    Sparsifying convex combination}]
 \lemlab{sparsify:c}%
 Let $\pp$ be a point in $\Re^d$, and let
 $\PS = \{ \pp_1, \ldots, \pp_m\} \subseteq \Re^d$ be a point
 set. Furthermore, assume that we are given $\pp$ as a convex
 combination of points of $\PS$. Then, one can compute a convex
 combination representation of $\pp$ that uses at most $d+1$ points of
 $\PS$. This takes $O(m d^3)$ time.
\end{lemma}

\begin{proof}
    Assume that we are given $\alpha_1,\ldots, \alpha_m \in [0,1]$,
    $\sum_i \alpha_i = 1$, such that
    $\pp = \sum_{i=1}^m \alpha_i \pp_i$, One can now sparsify the set
    $\PS$ so that it contains only $d+1$ points. Indeed, if any point
    of $\PS$ has zero coefficient in the representation of $\pp$ then
    it can be deleted. If there are more than $d+1$ points with
    non-zero coefficient, then pick $d+2$ of them, say
    $\pp_1, \ldots, \pp_{d+2}$. Computing their Radon decomposition,
    we get convex coefficients
    $\alpha_1, \ldots, \alpha_{k}, \beta_{k+1}, \ldots, \beta_{d+2}
    \in [0,1]$, such that $\sum_i \alpha_i =1$, $\sum_i \beta_i =1$,
    and
    $\sum_i \alpha_i \pp_i + \sum_{i=k+1}^{d+2} -\beta_i \pp_i =
    0$. Subtracting this equation from the current representation of
    $\pp$ (scaled with the appropriate constant) yields a
    representation of $p$ as a convex combination with more zero
    coefficients. This takes $O(d^3)$ time, and repeating it at most
    $m-d$ times results in the desired reduced set of size $d+1$. The
    process takes $O(d^3m)$ time overall.
\end{proof}

\begin{lemma}[{\normalfont \cite{ms-acp-10}: Merging logs}]
    \lemlab{double:rank}%
    Given $d+2$ sites $(\pp_1, \Log_1),$
    $\ldots,(\pp_{d+2}, \Log_{d+2})$ of rank $\rank$ in $\Re^d$, where
    the logs $\Log_i$ are disjoint, one can compute a site
    $(\pp, \Log)$ of rank $2\rank$ in $O( \rank d^5)$ time.
\end{lemma}

\begin{proof}
    Let $\pp$ be the Radon point of $\pp_1, \ldots, \pp_{d+2}$, and
    let the Radon partition be given by
    $\PS_1 = \{ \pp_1, \ldots, \pp_k\}$ and
    $\PS_2 = \{ \pp_{k+1},\ldots, \pp_{d+2}\}$. Picking the first set
    $\QS_i$ in the long $\Log_i$, for $i=1,\ldots, k$, results in a
    set $\cup \QS_i$ of size $O(d^2)$, whose convex hull contains
    $\CHX{ \PS_1}$. It also contains $\pp$, as
    $\pp \in \CHX{ \PS_1} \cap\CHX{\PS_2}$.

    Since every log contains $\rank$ sets, we can repeat this process
    $\rank$ times. we thus get $\rank$ disjoint sets, each of size
    $O(d^2)$, from the logs of the points of $\PS_1$, and each of
    their convex hulls contains $\pp$. Similarly, we get a similar log
    of size $\rank$ from $\PS_2$, and the union of the two logs is the
    desired new log of $\pp$ of rank $2\rank$. It is straightforward
    to compute the convex representation of $\pp$ in these new logs.

    The only issue is that each set $\QS$ in the new log has size
    $O(d^2)$.  We now sparsify every such set using
    \lemref{sparsify:c}.  This takes $O(d^5)$ time per set, and
    $O(d^5 \rank)$ time overall.
\end{proof}

\paragraph*{A recycling algorithm for computing a Tverberg point.}

The algorithm of Miller and Sheehy maintains a collection of
sites. Initially, it converts each input point $\pp \in \PS$ into a
site $(\pp, \{\pp\})$ of rank one. The algorithm then merges $d+2$
sites of rank $r$ into a site of rank $2r$, using
\lemref{double:rank}. Before the merge, the input logs use
$\rank(d+2) (d+1)$ points in total. After the merge, the new log uses
only $2 (d+1) \rank$ points. The algorithm recycles the remaining
$(d+2) (d+1)\rank - 2 (d+1) \rank = d(d+1)\rank$ points by reinserting
them into the collection as singleton sites of rank one.  When no
merges are available, the algorithm outputs the maximum-rank site as
the approximate Tverberg point.

\paragraph*{Analysis.}
For our purposes, we need a slightly different way of analyzing the
above algorithm of Miller and Sheehy \cite{ms-acp-10}, so we provide
the analysis in detail.

\begin{lemma}
    \lemlab{m:s}%
    Let $2^h$ be the rank of the output site. For $n \geq 8d(d+1)^2$,
    we have that
    \begin{equation*}
        \ceil{ \frac{n}{2(d+1)^2} }%
        \leq %
        2^h%
        \leq  %
        \frac{2n}{(d+1)^2}.
    \end{equation*}
\end{lemma}
\begin{proof}
    Consider the logs when the algorithm stopped. The number of points
    in all the logs is $n$. Since a site of rank $2^i$ has at most
    $(d+1)2^i$ points of $\PS$ in its log, and there are at most $d+1$
    sites of each rank, $n$ is at most $\sum_{i=0}^h (d+1)^2 2^i$. As
    such,
    \begin{equation}
        n%
        \leq %
        \sum_{i=0}^h (d+1)^2
        2^i
        =%
        (d+1)^2(2^{h+1} -1)%
        \iff%
        \frac{n}{(d+1)^2} + 1 \leq 2^{h+1}
        \implies%
        2^h \geq  \ceil{ \frac{n}{2(d+1)^2} }.
    \end{equation}

    As for the lower bound, consider the logs just before the output
    site was computed. There are $d+2$ sites, each of rank
    $2^{h-1}$. Each of their logs contains $(d+1)2^{h-1}$ points.
    Here, we use the lower bound on $n$, as a batch in a log of a site
    of rank smaller than (say) $2d$ can have fewer elements. However,
    under general position assumption, the way the above algorithm
    works, once the rank is sufficiently large (i.e., $> d+1$), it
    must be that each batch in the log indeed contains $d+1$
    points. Indeed, initially, a batch is of size one. If the site if
    of rank $2^i$, then a batch for it is formed by the union of $2^i$
    original batches (i.e., at least $2^i$ points) -- the algorithm
    then trims such a set to be of size $d+1$. Once a site gets to be
    of sufficiently high rank, under general position assumption, it
    is not going to lie on a $k$ dimensional subspace induced by $k+1$
    original points of $\PS$, for $k < d$ (indeed, just consider
    slightly perturbing the points -- the probability for this to
    happen is zero).  As such, we conclude that
    $n\geq (d+1)(d+2) 2^{h-1}$. This implies that
    \begin{math}
        2^{h} \leq 2n/ \bigl( (d+1)(d+2) \bigr) \leq 2n /(d+1)^2.
    \end{math}
\end{proof}

Every site in the collection is associated with a history tree that
describes how it was computed. Thus, the algorithm execution generates
(conceptually) a forest of such history trees, which is the
\emphi{history} of the computation.

\begin{lemma}
    \lemlab{max:rank}%
    Let $2^h$ be the maximum rank of a site computed by the
    algorithm. The total number of sites in the history that are of
    rank (exactly) $2^{h-i}$ is at most $(d+2)^{i+1} -1$.
\end{lemma}
\begin{proof}
    There are at most $T(0) = d+1$ nodes of rank $2^{h-0}$ in the
    history, and each has $d+2$ children of rank $2^{h-1}$. In
    addition, there might be $d+1$ trees in the history whose roots
    have rank $2^{h-1}$. Thus,
    \begin{equation*}
        T(1)%
        =%
        d+1 + (d+2)T(0) = (d+1)+ (d+2)(d+1)%
        =%
        (d+3)(d+1)
        =%
        (d+2)^2 - 1.
    \end{equation*}
    Assume that $T(i-1) = (d+2)^{i}-1$.  By induction, there are at
    most
    \begin{equation*}
        T(i)%
        =%
        d+1  + (d+2)T(i-1)%
        =%
        d+1 + (d+2)\bigl((d+2)^i - 1\bigr)
        =%
        (d+2)^{i+1} -1
    \end{equation*}
    nodes in the history of rank $2^{h-i}$.
\end{proof}

\begin{lemma}
    \lemlab{top:r:t}%
    Let $2^h$ be the maximum rank of a site computed by the
    algorithm. The total amount of work spent (directly) by the
    algorithm (throughout its execution) at nodes of rank
    $\geq 2^{h-i}$ is
    \begin{math}
        O( d^{4+i} n).
    \end{math}
    In particular, the overall running time of the algorithm is
    $O(d^4 n^{1 + \log_2 d})$.
\end{lemma}
\begin{proof}
    Bu \lemref{max:rank}, there are at most $(d+2)^{i+1}$ nodes of
    rank $2^{h-i}$ in the history.  By \lemref{m:s}, the rank of such
    a node is
    \begin{equation*}
        \rank_i%
        =%
        \frac{2^h}{2^i}%
        \leq%
        \frac{2n}{2^{i}(d+1)^2},
    \end{equation*}
    and the amount of work spent in computing the node is
    \begin{math}
        O( \rank_i d^6) =%
        O \pth{ d^4 n / 2^i }.
    \end{math}
    As such, the total work in the top
    $i$ ranks of the history is proportional to
    \begin{math}
        L_i = \sum_{j=0}^i (d+2)^j \cdot d^4 n /2^j = O( d^{4+i} n).
    \end{math}

    Since $h \leq \log_2 n$, the overall running time is
    $O( L_h) = O(d^{4 + \log_2n}n) = O(d^4 n^{1 + \log_2 d})$.
\end{proof}

\section{Improved Tverberg approximation %
   algorithms}
\seclab{improved}

\subsection{Projections in low dimensions}

\begin{lemma}
    \lemlab{t:3:d}%
    Let $\PS$ be a set of $n$ points in three dimensions. One can
    compute a Tverberg point of $\PS$ of depth $n/6$ (and the log
    realizing it) in $O(n \log n)$ time.
\end{lemma}
\begin{proof}
    Project the points of $\PS$ to two dimensions, and compute a
    Tverberg point $\pp$ and a partition for it of size $n/3$, using
    \thmref{birch}. Lifting from the plane back to the original space,
    the point $\pp$ lifts to a vertical line $\Line$, and every
    triangle lifts to a triangle that intersects $\Line$. Pick a point
    $\pq$ on this line which is the median of the intersections. Now,
    pair every triangle intersecting $\Line$ above $\pq$ with a
    triangle intersecting $\Line$ below it. This partitions $\PS$ into
    $n/6$ sets, each of size $6$, such that the convex hull of each
    sets contains $\pq$. Within each set, compute at most $4$ points
    whose convex hull contains $\pq$. These points yield the desired
    log for $\pq$ of rank $n/6$.
\end{proof}

\begin{remark}
    \lemref{t:3:d} seems innocent enough, but to the best of our
    knowledge, it is the best one can do in near-linear time in 3D.
    The only better approximation algorithm we are aware of is the one
    suggested by Tverberg's theorem.  It yields a point of Tverberg
    depth $n/4$, but its running time is $O(n^{13})$ (see
    \lemref{tverberg:exact:alg}).

    As observed by Mulzer and Werner \cite{mw-atplt-13}, one can
    repeat this projection mechanism.  Since Mulzer and Werner
    \cite{mw-atplt-13} bottom their recursion at dimension $1$, their
    algorithm computes a point of Tverberg depth $n/ 2^{d}$ (in three
    dimensions, depth $n/8$).  Applying this projection idea but
    bottoming at two dimensions, as above, yields a point of Tverberg
    depth $n/(3\cdot 2^{d-2})$.
\end{remark}

\begin{lemma}
    \lemlab{t:4:d}%
    Let $\PS$ be a set of $n$ points in four dimensions. One can
    compute a Tverberg point of $\PS$ of depth $n/9$ (and the log
    realizing it) in $O(n \log n)$ time.

    More generally, for $d$ even, and a set $\PS$ of $n$ points in
    $\Re^d$, one can compute a point of depth $n/3^{d/2}$ in
    $d^{O(1)} n \log n$ time. For $d$ odd, we get a point of depth
    $n/(2\cdot 3^{(d-1)/2})$.
\end{lemma}
\begin{proof}
    As mentioned above, the basic idea is due to Mulzer and
    Werner. Project the four-dimensional point set onto the plane
    spanned by the first two coordinates (i.e. ``eliminate'' the last
    two coordinates), and compute a $n/3$ centerpoint using
    \thmref{birch}. Translate the space so that this centerpoint lies
    at the origin. Now, consider each triangle in the original
    four-dimensional space. Each triangle intersects the
    two-dimensional subspace formed by the first two coordinates. Pick
    an intersection point from each lifted triangle. On this set of
    $n/3$ points, living in this two dimensional subspace, apply again
    the algorithm of \thmref{birch} to compute a Tverberg point of
    depth $(n/3)/3 = n/9$. The resulting centerpoint $\pp$ is now
    contained in $n/9$ triangles, where every vertex is contained in
    an original triangle of points. That is, $\pp$ has depth $n/9$,
    where each group consists of $9$ points in four dimensions. Now,
    sparsify each group into $5$ points whose convex hull contains
    $\pp$.

    The second part of the claim follows from applying the above
    argument repeatedly.
\end{proof}

\subsection{An improved quasi-polynomial algorithm}

The algorithm of Miller and Sheehy is expensive at the bottom of the
recursion tree, so we replace the bottom with a faster algorithm. We
use the following result of Mulzer and Werner.

\begin{theorem}[\cite{mw-atplt-13}]
    \thmlab{m:w}%
    Given a set $\PS$ of $n$ points in $\Re^d$, one can compute a site
    of rank $\geq n/4(d+1)^3$ (together with its log) in
    $d^{O(\log d)}n $ time.
\end{theorem}

Let $\delta \in (0,1)$ be a small constant. We modify the algorithm of
Miller and Sheehy by keeping all the singleton sites of rank one in a
buffer of \emphw{free} points. Initially, all points are in the
buffer.  Whenever the buffer contains at least $\delta n$ points, we
use the algorithm of \thmref{m:w} to compute a site of rank
$\rho \geq \delta n /8(d+1)^3$, where $\rho$ is a power of two.  (If
the computed rank is too large, we throw away entries from the log
until the rank reaches a power of two.) We insert this site into the
collection of sites maintained by the algorithm.  As in Miller and
Sheehy's algorithm, we repeatedly merge $d+2$ sites of the same rank
to get a site of double the rank.  In the process, the points thrown
out from the log are recycled into the buffer. Whenever the buffer
size exceeds $\delta n$, we compute a new site of rank $\rho$. This
process stops when no sites can be merged and the number of free
points is less than $\delta n$.

\begin{theorem}
    \thmlab{m:w:better}%
    Given a set $\PS$ of $n$ points in $\Re^d$, and a parameter
    $\delta \in (0, 1)$, one can compute a site of rank at least
    $\frac{(1-\delta)n}{2(d+1)^2}$ (together with its log) in
    $d^{O(\log (d/\delta))}n $ time.
\end{theorem}

\begin{proof}
    When the algorithm above stops, there are at least $(1-\delta)n$
    points in its logs.  Arguing as in \lemref{m:s}, the output site
    has rank
    \begin{equation*}
        2^h \geq \frac{(1-\delta)n}{2(d+1)^2}.
    \end{equation*}

    We now consider the running time.  The algorithm maintains only
    nodes of rank $\rho \geq 2^{h-H}$, where
    \begin{equation*}
        H%
        =%
        \ceil{\log_2 \frac{ 2^h}{\rho}}%
        \leq%
        \ceil{\log_2 \frac{ n/2(d+1)^2}{\delta n /8(d+1)^3}}%
        =%
        1 + \log_2 \frac{4(d+1)}{\delta}
        =%
        O( \log(d/\delta ) ).
    \end{equation*}
    The total work spent on merging nodes with these ranks is
    equivalent to the work in Miller and Sheehy's algorithm for such
    nodes.  By \lemref{top:r:t}, the total work performed is
    $O(d^{4+H} n)$.

    As for the work associated with the buffer, observe that
    \thmref{m:w} is invoked $(d+2)^{H+1}$ times (this is the number of
    nodes in the history of rank $2^{h-H}$). Each invocation takes
    $d^{O(\log d)} n$ time, so the total running time of the algorithm
    is
    \begin{equation*}
        d^{O(H)} n + (d+1)^{H+1} d^{O(\log d)} n = d^{O(\log(d/\delta))} n.
    \end{equation*}
\end{proof}

\subsection{A strongly polynomial algorithm}
\seclab{strong}

\begin{lemma}
    \lemlab{tverberg:1}%
    Let $\PS$ be a set of $n$ points in $\Re^d$. For
    $N = \Theta( d^2 \log d)$, consider a random coloring of $\PS$ by
    $k=\floor{n/N}$ colors, and let $\{ \PS_1, \ldots, \PS_k \}$ be
    the resulting partition of $\PS$. With probability
    $\geq 1- k/d^{O(d)}$, this partition is a Tverberg partition of
    $\PS$ -- that is, $\cap_i \CHX{ \PS_i} \neq \emptyset$.
\end{lemma}
\begin{proof}
    Let $\eps = 1/(d+1)$.  The \VC dimension of halfspaces in $\Re^d$
    is $d+1$ by Radon's theorem. By the $\eps$-net theorem
    \cite{hw-ensrq-87,h-gaa-11}, a sample from $\PS$ of size
    \begin{equation*}
        N = \ceil{ \frac{8(d+1)}{\eps} \log (16(d+1)) }
        =%
        \Theta(d^2 \log d)
    \end{equation*}
    is an $\eps$-net with probability
    $\geq 1 - 4\bigl(16(d+1)\bigr)^{-2(d+1)}$. Thus
    $\PS_1, \ldots, \PS_k$ are all $\eps$-nets with the probability
    stated in the lemma.

    Now, consider the centerpoint $\pq$ of $\PS$. We claim that
    $\pq \in \CHX{\PS_i}$, for all $i$. Indeed, assume otherwise, so
    that there is a separating hyperplane between $\pq$ and some
    $\PS_i$. Then the halfspace induced by this hyperplane that
    contains $\pq$ also contains at least $\eps n$ points of $\PS$,
    because $\pq$ is a centerpoint of $\PS$. But this contradicts that
    $\PS_i$ is an $\eps$-net for halfspaces.
\end{proof}

\begin{lemma}
    \lemlab{partition}%
    Let $\PS$ be a set of $n$ points in $\Re^d$. For
    $N = O( d^3 \log d)$, consider a random coloring of $\PS$ by
    $k=\floor{n/N} = \Omega( \frac{n}{d^3 \log d})$ colors, and let
    $\{ \PS_1, \ldots, \PS_k \}$ be the resulting partition of $\PS$.
    One can compute, in $O(d^{7} \log^{6} d )$ time, a Tverberg point
    that lies in $\bigcap_i \CHX{\PS_i}$.  This algorithm succeeds
    with probability $\geq 1- k/d^{O(d)}$.
\end{lemma}
\begin{proof}
    Compute a $1/(2 d^2)$-centerpoint $\pq$ for $\PS$ using the
    algorithm of Har-Peled and Jones \cite{hj-jcps-19}, and repeat the
    argument of \lemref{tverberg:1} with $\eps = 1/(2d^2)$. With the
    desired probability, $\pq \in \CHX{\PS_i}$ for all $i$.
\end{proof}

\subsection{A weakly polynomial algorithm}

Let $\TLPY{n}{m}$ be the time to solve an \LP with $n$ variables and
$m$ constraints. Using interior point methods, one can solve such
\LP{}s in weakly polynomial time. The fastest known method runs in
$\Ow( mn + n^{3})$ \cite{blss-stdlp-20} or $\Ow( mn^{3/2})$
\cite{ls-eimfa-15}, where $\Ow$ hides polylogarithmic terms that
depends on $n, m$, and the width of the input numbers.

\newcommand{\Eps}{\mathcal{E}}%
\begin{remark}
    \remlab{width}%
    The \emph{error} of the \LP solver, see \cite{blss-stdlp-20}, for
    a prescribed parameter $\eps$, is the distance of the computed
    solution from an optimal one. Specifically, let $R$ be the maximum
    absolute value of any number in the given instance.  In
    $O( (nm + n^3) \log^{O(1)} n \log (n/\eps) )$ time, the \LP solver
    can find an assignment to the variables which is $\Eps$ close to
    complying with the \LP constraints, where $\Eps \leq \eps n m R$
    \cite{blss-stdlp-20}.  That is, the \LP solver can get arbitrarily
    close to a true solution. This is sufficient to compute an exact
    solution in polynomial time if the input is made out of integer
    numbers with polynomially bounded values, as the running time then
    depends on the number of bits used to encode the input. We use
    $\Ow(\cdot)$ to denote such \emphi{weakly polynomial} running
    time.
\end{remark}

\begin{lemma}[\Caratheodory via \LP]
    \lemlab{c:lp}%
    Given a set $\PS$ of $n$ points in $\Re^d$ and query point $\pq$,
    one can decide if $\pq \in \CHX{\PS}$, and if so, output a convex
    combination of $d+1$ points of $\PS$ that is equal to $\pq$. The
    running time of this algorithm is
    $O( \TLPY{n}{n+d} + nd^3) = \Ow(nd + nd^3)$.

    Alternatively, one can compute such a point in
    $2^{O(\sqrt{d \log d})} + O(d^2 n)$ time.
\end{lemma}
\begin{proof}
    We write the natural \LP for representing $\pq$ as a point in the
    interior of $\CHX{\PS}$. This \LP has $n$ variables (one for each
    point), and $n+2d + 2$ constraints (all variables are positive,
    the points sum to $\pq$, and the coefficients sum to $1$), where
    an equality constraint counts as two constraints. If the \LP
    computes a solution, then we can sparsify it using
    \lemref{sparsify:c}.

    The alternative algorithm writes an \LP to find a hyperplane
    separating $\pq$ from $\PS$. First, it tries to find a hyperplane
    which is vertically below $\PS$ and above $\pq$. This \LP has $d$
    variables, and it can be solved in
    $O( 2^{\sqrt{d \log d}} + d^2 n )$ time
    \cite{s-sdlpc-91,c-lvali-95,msw-sblp-96}. If there is no such
    separating hyperplane, then the algorithm provides $d$ points
    whose convex hull lies below $\pq$ and tries again. This time, it
    computes a separating hyperplane below $\pq$ and above $\PS$.
    Again, if no such separating hyperplane exists, then the algorithm
    returns $d$ points whose convex hull lies above $\pq$. The union
    of the computed point sets above and below $\pq$ contains at most
    $2d$ points. Now, write the natural \LP with $2d$ variables as
    above, and solve it using linear-time \LP solvers in low
    dimensions.  This gives the desired representation.
\end{proof}

\begin{remark}
    Rolnick and \Soberon \cite{rs-aatt-16} achieved a result similar
    to that of \lemref{c:lp}, but they used binary search to find the
    $d+1$ coefficients in the minimal representation. Thus their
    running time is (slightly worse) $O( \TLPY{n}{n+d} d \log n ) $.
\end{remark}

Using \lemref{tverberg:1}, we get the following.

\begin{lemma}
    \lemlab{tverberg:1:lp}%
    Let $\PS$ be a set of $n$ points in $\Re^d$. For
    $N = O( d^2 \log d)$, one can compute a site $\pq$ of rank
    $k = n/O(d^2 \log d)$ in $\Ow(n^{5/2} + nd^3)$ time.  The
    algorithm succeeds with probability $\geq 1- k/d^{O(d)}$.
\end{lemma}
\begin{proof}
    We compute a Tverberg partition using \lemref{tverberg:1}. Next we
    write an \LP for computing an intersection point $\pq$ that lies
    in the interior of the $k = n/O(d^2 \log d)$ sets. This \LP has
    $O(n+d)$ constraints and $n$ variables, and it can be solved in
    $\Ow(n^{5/2})$ time \cite{ls-eimfa-15}. We then sparsify the
    representations over each of the $k$ sets. This requires
    $O( d^2 \log d \cdot d^3 )$ time per set and hence $O(n d^3)$ time
    overall. The result is a site $\pq$ with a log of rank $k$, as
    desired.
\end{proof}

\begin{theorem}
    \thmlab{m:w:better:lp}%
    Given a set $\PS$ of $n$ points in $\Re^d$, and a parameter
    $\delta \in (0,1)$, one can compute a site of rank at least
    $\frac{(1-\delta)n}{2(d+1)^2}$ (together with its log) in
    $d^{O(\log \log (d/\delta))}\Ow( n^{5/2}) $ time.
\end{theorem}
\begin{proof}
    We modify the algorithm of \thmref{m:w:better} to use
    \lemref{tverberg:1:lp} to compute a Tverberg point on the points
    in the buffer. Since the gap between the top rank of the recursion
    tree and rank computed by \lemref{tverberg:1:lp} is
    $O( \log (d/\delta))$, it follows that the algorithm uses only the
    top $O( \log \log (d/\delta))$ levels of the recursion tree, and
    the result follows.
\end{proof}

\newcommand{\Sample}{\Mh{R}}%
\subsection{Faster approximation in low dimensions}

\begin{lemma}
    \lemlab{t:low:dim}%
    Let $\PS$ be a set of $n$ points in $\Re^d$, and let
    $\delta \in (0,1)$ be some parameter. One can compute a Tverberg
    point of depth $\geq (1-\delta)n/d(d+1)$, together with its log,
    in
    \begin{math}
        O(d^{O(d)} /\delta^{2(d-1)} + 2^{O(\sqrt{d\log d})} n + d^2 n
        \log^2 n )
    \end{math}
    time.  The algorithm succeeds with probability close to one.
\end{lemma}

\begin{proof}
    Let $\eps=1/2(d+1)$, and let $\Sample$ be a random sample from
    $\PS$ of size
    \begin{equation*}
        m = O(d \eps^{-1} \delta^{-2} \log \eps^{-1} ) = O( d^2
        \delta^{-2} \log d ).
    \end{equation*}
    This sample is a $(\eps,\delta)$-relative approximation for $\PS$
    with probability close to one \cite{hs-rag-11}. Compute a
    centerpoint $\pc$ for the sample using ``brute force'' in
    $O( m^{d-1})$ time \cite{c-oramt-04}. Now, repeatedly use the
    algorithmic version of \Caratheodory's theorem (\lemref{c:lp}) to
    extract a simplex that contains $\pc$. One can repeat this process
    $\floor{n/d(d+1)}$ times before getting stuck, since every simplex
    contains at most $d$ points in any halfspace passing through
    $\pc$. Naively, the running time of this algorithm is
    $2^{O(\sqrt{d \log d})} n^2$.

    However, one can do better.  The number of simplices extracted by
    the above algorithm is $L= \floor{n/d(d+1)}$.  For $i=1,\ldots L$,
    let $b_i = \floor{n/(d+1)} -d(i-1)$ be a lower bound on the Tukey
    depth of $\pc$ at the beginning of the $i$\th iteration of the
    above extraction algorithm. At this moment, the point set has
    $n_{i} = n -(d+1)(i-1) \geq n/2$ points. Hence the relative Tukey
    depth of $\pc$ is
    \begin{math}
        \eps_{L-i} %
        =%
        { b_{L-i}}/{n_{L-i}}.%
    \end{math}
    In particular, an $\eps_i$-net $\Sample_i$ for halfspaces has size
    \begin{math}
        r_i%
        =%
        O(\frac{d}{\eps_i} \log \frac{1}{\eps_i} ).
    \end{math}
    If $r_{i}$ is larger than the number of remaining points, then the
    sample consists of the remaining points of $\PS$. The convex hull
    of such a sample contains $\pc$ with probability close to one, so
    one can apply \lemref{c:lp} to $\Sample_i$ to get a simplex that
    contains $\PS$ (if the sample fails, then the algorithm
    resamples). The algorithm adds the simplex to the output log,
    removes its vertices from $\PS$, and repeats.

    Since the algorithm invokes \lemref{c:lp} $O(n/d^2)$ times, the
    running time is bounded by
    \begin{math}
        2^{O(\sqrt{d\log d})} n/d^2 + O(d^2 \sum_i r_i).
    \end{math}
    Since $n_{L-i} \geq n/2$ and $b_{L-i} \geq i \cdot d$, we have
    that $\eps_{L-i} \geq (i d/(n/2) = 2id/n$. Therefore,
    \begin{equation*}
        \sum_i r_i%
        =%
        O(n ) + \sum_{i=1}^{L-1} r^{}_{L-i} %
        =%
        O(n  ) + \sum_{i=1}^{L-1} O \pth{ \frac{d n}{2id} \log
           \frac{n}{id}}%
        =%
        O( n \log^2 n).
    \end{equation*}
\end{proof}

\paragraph*{Acknowledgments.}

The authors thank Timothy Chan, Wolfgang Mulzer, David Rolnick, and
Pablo Soberon-Bravo for providing useful references.

\BibTexMode{%
   \SoCGVer{%
      \bibliographystyle{plain}%
   }%
   \NotSoCGVer{%
      \bibliographystyle{alpha}%
   }%
   \bibliography{triangle_cover} }

\BibLatexMode{\printbibliography}

\end{document}